\documentclass[12pt, draftclsnofoot, onecolumn]{IEEEtran}
 \usepackage{amsmath,amssymb}
 \usepackage{subfigure}
 \usepackage{graphicx,graphics,color,psfrag}
 \usepackage{cite,balance}
 \usepackage{algorithm}
 \usepackage{accents}
 \usepackage{amsthm}
 \usepackage{bm}
 \usepackage{url}
\usepackage{algorithmic}
 \usepackage[english]{babel}
 \usepackage{multirow}
 \usepackage{enumerate}
 \usepackage{cases}
 \usepackage{stfloats}
 \usepackage{dsfont}
 \usepackage{color,soul}
 \usepackage{amsfonts}
 \usepackage{cite,graphicx,amsmath,amssymb}
 \usepackage{subfigure}
 \usepackage{fancyhdr}
 \usepackage{hhline}
 \usepackage{graphicx,graphics}
 \usepackage{array,color}
 \usepackage{amsmath}
 \usepackage{amsthm}
 \usepackage[colorlinks,linkcolor=blue,citecolor=blue]{hyperref}

 \theoremstyle{definition}
 \newtheorem{remark}{Remark}
 \newtheorem{proposition}{Proposition}

\begin{document}

\title{Bare Demo of IEEEtran.cls for\\ IEEE \textsc{Transactions on Magnetics}}


\title{Composite Signalling for DFRC: \\
Dedicated Probing Signal or Not?}

\author{Li Chen, Fan Liu,~\IEEEmembership{Member,~IEEE,} Jun Liu,~\IEEEmembership{Senior Member,~IEEE,} and Christos Masouros,~\IEEEmembership{Senior Member,~IEEE} 
\thanks{


L. Chen and J. Liu are with the Department of Electronic Engineering and Information Science, University of Science and Technology of China. 
(e-mail: \{chenli87, junliu\}@ustc.edu.cn).

F. Liu and C. Masouros are with the Department of Electronic and Electrical Engineering, University College London, London, WC1E 7JE, UK (e-mail: fan.liu@ucl.ac.uk, chris.masouros@ieee.org).

}}



\maketitle

\begin{abstract}
Dual-functional radar-communication (DFRC) is a promising new solution to simultaneously probe the radar target and transmit information in wireless networks. In this paper, we study the joint optimization of transmit and receive beamforming for the DFRC system. Specifically, the signal to interference plus noise ratio (SINR) of the radar is maximized under the SINR constraints of the communication user (CU), which characterizes the optimal tradeoff between radar and communication. 
In addition to simply using the  communication signal for target probing, we further consider to exploit dedicated probing signals to enhance the radar sensing performance.
We commence by studying the single-CU scenario, where a closed-form solution to the beamforming design problem is provided. It is then proved that a dedicated radar probing signal is not needed.
As a further step,  we consider a more complicated multi-CU scenario, where the beamforming design is formulated as a non-convex quadratically constrained quadratic programming. The optimal solutions are obtained by applying semidefinite relaxation with guaranteed rank-1 property. It is shown that under the multi-CU scenario, the dedicated probing signal should be employed to improve the radar performance at the cost of implementing an additional interference cancellation at the CU.
Finally, the numerical simulations are provided to verify the effectiveness of the proposed algorithm.
\end{abstract}


\section{Introduction}
Communication and radar spectrum sharing (CRSS) has recently drawn significant attention due to the scarcity of the commercial wireless spectrum. For instance, the millimeter wave (mmWave) band is occupied by variety of radars \cite{7786130}, and has also been assigned as a new licensed band to the 5G network \cite{7000981}. It is well-recognized that communication and radar signals have some common features in their waveforms. Although their purposes are dramatically different, it is feasible to use one type of signal for the other type’s purpose. Nevertheless, the use of radar (communication) signals for communication (radar) functionalities, introduces a number of challenges \cite{8828023,8352726,8737000}. To address these challenges, the research of dual-functional radar-communication (DFRC) is well-underway \cite{8828016,8642926,9005192}.

In general,  the aim of the DFRC is to implement both communication and radar functionalities on the same hardware platform. Based on information theory, the work in \cite{7279172} unified the radar and communication performance metric and discussed the performance bounds of the DFRC system. 
Furthermore, the weighted
sum of the estimation and communication rates was analyzed as the performance metrics in the DFRC system \cite{7855671}. 
By leveraging the simple time-division scheme,  the radar and the communication signals can be transmitted within different time slots, which avoids the mutual interference \cite{5483108}. To exploit the favorable time-frequency decoupling property of the orthogonal frequency division multiplexing (OFDM) waveform, the OFDM communication signal was adopted for target detection, where the range and Doppler processing are independent with each other \cite{5776640}.  
From a signal processing perspective, the implementation of mmWave DFRC systems was fully studied in \cite{8828030}. The unimodular signal design was discussed in \cite{9119137} for DFRC architecture, where the information of downlink communication was modulated via the ambiguity function (AF) sidelobe nulling in the prescribed range-Doppler cells.

Beamforming design is essential to improve the performance of the DFRC signal processing in the spatial domain, which has been widely studied in the literature. Aiming to realize the dual functionalities,  the  work  of  \cite{7347464}  designed  a  transmit  beampattern for multiple input multiple output (MIMO) radar with the communication information being embedded into the sidelobes of the radar beampattern. Considering both the separated and the shared antenna deployments, a series of optimization-based transmit beamforming approaches for the DFRC system were studied in \cite{8288677}, where the communication signal was exploited for target detection. By imposing the constraints of the radar waveform similarity and the constant modulus, the interference of the multiple communication users (CUs)  was suppressed to improve the communication performance in \cite{8386661}. 
Based on IEEE 802.11ad wireless local area network (WLAN) protocol, a joint waveform for automotive radar and a potential mmWave vehicular communication system were proposed in \cite{8114253}. The work of \cite{8309274} further studied the feasibility of an opportunistic radar, which exploited the probing signals transmitted during the sector level sweep of the IEEE 802.11ad beamforming training protocol.
In order to reduce the hardware complexity and the associated costs incurred in the mmWave massive MIMO system, a hybrid analog-digital beamforming structure was proposed for the DFRC transmission in \cite{8999605}. 

It is worth pointing out that all the aforementioned works on the DFRC beamforming design focused on formulating a desired transmit beampattern without considering the receive beamformimg. To the best of our knowledge, the joint optimization of the transmit and receive beamforming under the communication constraints has never been studied for the DFRC system before, despite the fact that the joint optimization of transmit and receive beamforming for the MIMO radar system has been extensively investigated in the recent literature \cite{4840496,6649991,7450660,7762192,8239836,9034082, 8826594}.
To detect an extended target, the joint optimization of waveforms and receiving filters in the MIMO radar was considered in \cite{4840496}. In order to guarantee  the constant modulus and similarity properties of the radar waveforms, numerous approaches were provided to maximize the signal to interference plus noise (SINR) of the radar, e.g., the sequential optimization algorithms (SOAs) in \cite{6649991}, the successive quadratically constrained quadratic programming (QCQP) refinement method in \cite{7450660}, the block coordinate descent (BCD) framework in \cite{7762192}, and the general majorization-minimization (MM) framework in \cite{8239836}. To detect multiple targets, the joint optimization of waveforms and receiving filters in the MIMO radar was further studied in \cite{9034082}. The problem of beampattern synthesis with sidelobe control was studied in \cite{8826594} using constant modulus weights.
The jointly design of the transmit and receive beamforming was provided in \cite{6656878} based on  \emph{a priori} information on the locations of target
and interferences.

All these works improved the receive SINR of the echo signal based on the prior information of the target and clutter. In contrast, for the DFRC system, 
there exist both probing signal and communication signal, which are coupled together with each other for
drastically different purposes. As a consequence, the known MIMO radar-only designs are inapplicable for the latter. 
In this paper, we study the joint optimization of transmit and receive beamforming for the DFRC system. Specifically, the SINR of the radar is maximized under the SINR constraints of the CUs. Depending on the component of the DFRC transmit signal, we consider both the non-dedicated probing signal case and the dedicated probing signal case. 
For the non-dedicated probing signal case, the DFRC transmit signal is only composed of the communication signal of the CUs. For the dedicated probing signal case, besides the communication signal, the dedicated probing signal is added to the DFRC transmit signal to improve the radar performance.
Under the single CU scenario,  closed-form solutions of the optimized beamforming are provided for both cases. And it can be proved that there is no need to employ dedicated probing signal for the single CU scenario. On top of that, we consider a more complicated scenario with multiple CUs. For the non-dedicated probing signal case, the beamforming design is formulated as a non-convex quadratically constrained quadratic programming (QCQP), and we show that the globally optimal solution can be obtained by applying semidefinite relaxation (SDR) with rank-1 property. For the dedicated probing signal case, the rank-1 property after applying SDR can be also proved, and the corresponding optimal solution shows that the dedicated probing signal should be employed to improve the SINR of the radar receiver. The main contributions of this paper can be summarized as follows.

\begin{itemize}
    \item \textbf{Transmit-receive DFRC beamforming}:  We provide the solutions  of the joint transmit-receive beamforming optimization to maximize the SINR of the radar under the SINR constraints of the CUs. For the single CU scenario, the closed-form solutions of the optimized beamforming are derived. For the multiple CUs scenario, the optimal solutions are obtained by applying SDR with rank-1 property.
    
    \item \textbf{Dual-functional performance tradoff}: The optimal performance tradeoff between the radar and the communication is characterized in terms  of SINR. Compared to the time-sharing scheme between the radar and the communication, the joint optimization of transmit and receive beamforming yields a more favorable tradeoff performance. 
    
    \item \textbf{Dedicated radar probing signal or not}: Compared to the existing works only exploiting communication signals for radar functionality, we consider the use of a dedicated radar probing signal to improve the radar performance. For the single CU scenario, our analysis shows that there is no need of dedicated probing signal. For the multi-CU scenario, on the other hand, it is beneficial  to  employ  the  dedicated  probing  signal.
    
\end{itemize}

The remainder of the paper is organized as follows. Section \uppercase\expandafter{\romannumeral2} presents the system model. Section \uppercase\expandafter{\romannumeral3} studies a simplified single CU scenario. The study is further extended to a more complicated multiple CUs scenario in Section \uppercase\expandafter{\romannumeral4}. 
Simulation results are provided in Section \uppercase\expandafter{\romannumeral5}, followed by concluding remarks in Section \uppercase\expandafter{\romannumeral6}.

\textbf{Notation}: We use boldface lowercase letter to denote column vectors, and boldface uppercase letters to denote matrices. Superscripts $(\cdot)^H$ and $(\cdot)^T$ stand for Hermitian transpose and transpose, respectively.
$\text{tr}(\cdot)$ and $\text{rank}(\cdot)$ represent the trace operation and the rank operator, respectively.  $\mathcal{C}^{m \times n}$ is the set of complex-valued $m \times n$ matrices.
$x \sim \mathcal{CN}\left( {a,b} \right)$ means that $x$ obeys a complex Gaussian distribution with mean $a$ and covariance $b$.
$\text{E}(\cdot)$ denotes the statistical expectation. $\|\bf{x}\|$ denotes the Euclidean norm of a complex vector $\bf{x}$.

\section{System Model}
As illustrated in Fig. \ref{system}, we consider a DFRC MIMO system, which simultaneously probes the radar target and transmits information to the CUs. To be specific, it is composed of a DFRC base station (BS) with $N_t$ transmit antennas and $N_r$ receive antennas, $K$ single-antenna CUs indexed by $k \in \left\{ {1, \cdots ,K} \right\}$.

\begin{figure}[t]
\centering
  \includegraphics[width=0.5\textwidth]{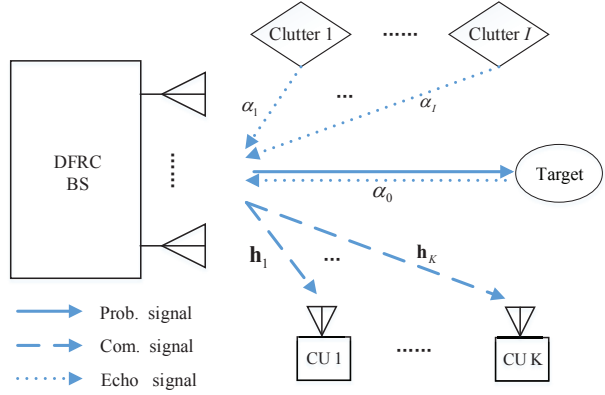}
\caption{System model of a DFRC MIMO system}
\label{system}
\end{figure}

Suppose there is a target and $I$  signal-dependent interference sources indexed by $i \in \left\{ {1, \cdots ,I} \right\}$. The target is located at angle $\theta_0$ and the interference sources are located at angle $\theta_i, i \in \{1,\cdots,I\}$. Given the transmit signal ${\bf{x}} \in {\mathcal{C}^{{N_t} \times 1}}$, the received signal of the radar receiver is

\begin{equation}
\begin{aligned}
 {{\bf{y}}_0} &= {\alpha _0}{{\bf{a}}_r}\left( {{\theta _0}} \right){\bf{a}}_t^T\left( {{\theta _0}} \right){\bf{x}} + \sum\limits_{i = 1}^I {{\alpha _i}{{\bf{a}}_r}\left( {{\theta _i}} \right){\bf{a}}_t^T\left( {{\theta _i}} \right){\bf{x}}}  + {{\bf{z}}_0}\\
 &={\alpha _0}{\bf{A}}\left( {{\theta _0}} \right){\bf{x}} + \sum\limits_{i = 1}^I {{\alpha _i}{\bf{A}}\left( {{\theta _i}} \right){\bf{x}}}  + {{\bf{z}}_0},
 \end{aligned}
\end{equation}
where ${\alpha _0}$ and ${\alpha _i}$ are the complex amplitudes of the target and the $i$-th interference source, respectively,  ${\bf{a}}_t^{}\left( \theta  \right) = {[ {1,{e^{ - j2\pi {\Delta _t}\sin \theta }}, \cdots ,{e^{ - j2\pi \left( {{N_t} - 1} \right){\Delta _t}\sin \theta }}} ]^T}$ and ${\bf{a}}_r^{}\left( \theta  \right) = {[ {1,{e^{ - j2\pi {\Delta _r}\sin \theta }}, \cdots ,{e^{ - j2\pi \left( {{N_r} - 1} \right){\Delta _r}\sin \theta }}} ]^T}$ with
${\Delta _t}$ and ${\Delta _r}$
being the spacing between adjacent antennas normalized by the wavelength, respectively, and ${{\bf{z}}_0} \in {{\cal C}^{{N_r} \times 1}}$ is the additive white Gaussian noise (AWGN) with each element subjects to ${\cal C}{\cal N}\left( {0,1} \right)$. The symbol index is omitted for simplicity.
Then, the output of the radar receiver is

\begin{equation}
\label{eq2}
\begin{aligned}
    r &= {{\bf{w}}^H}{{\bf{y}}_0} \\
    &= {\alpha _0}{{\bf{w}}^H}{\bf{A}}\left( {{\theta _0}} \right){\bf{x}} + {{\bf{w}}^H}\sum\limits_{i \in {\cal I}}^{} {{\alpha _i}{\bf{A}}\left( {{\theta _i}} \right){\bf{x}}}  + {{\bf{w}}^H}{{\bf{z}}_0},
\end{aligned}
\end{equation}
where ${\bf{w}} \in {{\cal C}^{{N_r} \times 1}}$ is the receive beamforming vector for SINR maximization.

Further, given the transmit signal ${\bf{x}} \in {\mathcal{C}^{{N_t} \times 1}}$, the received signal of the CU $k$ is

\begin{equation}
\label{eq3}
{y_k} = {\bf{h}}_k^H{\bf{x}} + {z_k},
\end{equation}
where ${{\bf{h}}_k} \in {{\cal C}^{{N_t} \times 1}}$ is the multiple input single output (MISO) channel vector between the DFRC BS and the CU $k$, and ${z_k} \sim {\cal C}{\cal N}\left( {0,1} \right)$ is the AWGN of the CU $k$. 

In this paper, we consider two cases according to the component of the DFRC BS's transmit signal ${\bf{x}}$.

\begin{itemize}
\item \textbf{Case 1} (Non-dedicated probing signal): In this case, the transmit signal of the DFRC BS is only composed of the communication signals of the CUs. That is 
\begin{equation}
\label{eqa4}
{\bf{x}} = \sum\limits_{k = 1}^K {{{\bf{x}}_k}},    
\end{equation}
where ${{\bf{x}}_k}$ is the communication signals of the CU $k$ and the radar functionality is realized by the sum of the CUs' communication signal.

\item \textbf{Case 2} (Dedicated probing signal): In this case, the transmit signal of the DFRC BS is composed of both the communication signal of the CUs and the dedicated probing signal.
That is 
\begin{equation}
\label{eqa5}
{\bf{x}} = \sum\limits_{k = 1}^K {{{\bf{x}}_k}}  + {{\bf{x}}_0},     
\end{equation}
where ${{\bf{x}}_0}$ is the dedicated probing signal to enhance the radar performance.

\end{itemize}

In addition, we impose the following two assumptions in this paper. 1) For the radar function, the angles of the target $\theta_0$ and the interference $\left\{ {{\theta _i}} \right\}$  are assumed to be known to the DFRC BS. 2) For the communication function, the channel 
is assumed to be known to the DFRC BS, and the dedicated probing signal ${{\bf{x}}_0}$ is pseudo-random and assumed to be known in prior to the CUs.

\section{Single CU Scenario}
In this section, we consider a simplified scenario with single CU in the network. The beamforming design of the DFRC BS is discussed for the non-dedicated probing signal case, and the closed-form solution is provided. Then, the dedicated probing signal case is studied, and it can be proved that there is no need of dedicated probing signal with single CU.

\subsection{Non-dedicated Probing Signal Case}
For the non-dedicated probing signal case, the transmit signal of the DFRC BS in (\ref{eqa4})
with single CU can be rewritten as

\begin{equation}
{\bf{x}} = {\bf{u}}s,
\end{equation}
where ${\bf{u}} \in {{\cal C}^{{N_t} \times 1}}$ and $s \in {\cal C}{\cal N}\left( {0,1} \right)$ are the beamforming vector and the information symbol of the CU, respectively. 

According to the output in (\ref{eq2}), the SINR of the radar receiver can be expressed as

\begin{equation}
\label{eq5}
\begin{aligned}
\gamma _R^{\left( \text{I} \right)} &= \frac{{{{\left| {{\alpha _0}{{\bf{w}}^H}{\bf{A}}\left( {{\theta _0}} \right){\bf{x}}} \right|}^2}}}{{{\rm{E}}\left[ {{{\left| {{{\bf{w}}^H}\sum\limits_{i=1}^{I} {{\alpha _i}{\bf{A}}\left( {{\theta _i}} \right){\bf{x}}} } \right|}^2}} \right] + {{\bf{w}}^H}{\bf{w}}}} \\
&= \frac{{{{\left| {{\alpha _0}} \right|}^2}{{\left| {{{\bf{w}}^H}{\bf{A}}\left( {{\theta _0}} \right){\bf{x}}} \right|}^2}}}{{{{\bf{w}}^H}\left[ {\sum\limits_{i = 1}^I {{{\left| {{\alpha _i}} \right|}^2}{\bf{A}}\left( {{\theta _i}} \right){\bf{uu}}_{}^H{{\bf{A}}^H}\left( {{\theta _i}} \right)}  + {\bf{I}}} \right]{\bf{w}}}}.
\end{aligned}
\end{equation}
And the output SINR of the radar receiver depends on the choice of the receive beamforming vector ${\bf{w}}$. The design of ${\bf{w}}$ can be expressed as

\begin{equation}
 \begin{array}{*{20}{c}}
{\mathop {\max }\limits_{\bf{w}} }&{\dfrac{{{{\left| {{{\bf{w}}^H}{\bf{A}}\left( {{\theta _0}} \right){\bf{x}}} \right|}^2}}}{{{{\bf{w}}^H}\left[ {\sum\limits_{i = 1}^I {{{\left| {{\alpha _i}} \right|}^2}{\bf{A}}\left( {{\theta _i}} \right){\bf{uu}}_{}^H{{\bf{A}}^H}\left( {{\theta _i}} \right)}  + {\bf{I}}} \right]{\bf{w}}}}}
\end{array},
\end{equation}
which is equivalent to the well-know minimum variance distortionless response (MVDR) problem, and its solution can be given by \cite{4840496}

\begin{equation}
\label{eq7}
{{\bf{w}}^*} = \frac{{{{\bf{\Sigma }}_1}{{\left( {\bf{u}} \right)}^{ - 1}}{\bf{A}}\left( {{\theta _0}} \right){\bf{x}}}}{{{\bf{x}}_{}^H{{\bf{A}}^H}\left( {{\theta _0}} \right){{\bf{\Sigma }}_1}{{\left( {\bf{u}} \right)}^{ - 1}}{\bf{A}}\left( {{\theta _0}} \right){\bf{x}}}},
\end{equation}
where 

\begin{equation}
{\bf{\Sigma }}_1\left( {\bf{u}} \right) = \left[ {\sum\limits_{i = 1}^I {{{\left| {{\alpha _i}} \right|}^2}{\bf{A}}\left( {{\theta _i}} \right){\bf{uu}}_{}^H{{\bf{A}}^H}\left( {{\theta _i}} \right)}  + {\bf{I}}} \right].
\end{equation}

Substituting (\ref{eq7}) into (\ref{eq5}), the SINR of the radar receiver can be calculated as

\begin{equation}
\gamma _R^{\left( \text{I} \right)} = {\bf{x}}_{}^H{\bf{\Phi }}_1\left( {\bf{u}} \right){\bf{x}},
\end{equation}
where 

\begin{equation}
\label{eq10}
{\bf{\Phi }}_1\left( {\bf{u}} \right) = {\left| {{\alpha _0}} \right|^2}{{\bf{A}}^H}\left( {{\theta _0}} \right){\bf{\Sigma }}_1{\left( {\bf{u}} \right)^{ - 1}}{\bf{A}}\left( {{\theta _0}} \right).
\end{equation}
And the average SINR of the radar receiver can be given by

\begin{equation}
\label{eq11}
\bar \gamma _R^{\left( \text{I} \right)} = {\rm{E}}\left[ {{\bf{x}}_{}^H{\bf{\Phi }}_1\left( {\bf{u}} \right){\bf{x}}} \right] = {\bf{u}}_{}^H{\bf{\Phi }}_1\left( {\bf{u}} \right){\bf{u}}.
\end{equation}

For the CU $k$, the received signal in (\ref{eq3}) can be rewritten as

\begin{equation}
{y_k} = {{\bf{h}}^H}{\bf{u}}s + {z_k},
\end{equation}
where ${\bf{h}} \in {{\cal C}^{{N_t} \times 1}}$ is the MISO channel vector between the DFRC BS and the CU. And
the average SINR of the CU can be calculated as

\begin{equation}
\label{eq12}
{\bar \gamma _C} = {\left| {{{\bf{h}}^H}{\bf{u}}} \right|^2},
\end{equation}
Then, we consider the beamforming optimization problem that maximizes the SINR of the radar receiver and satisfies the SINR of the CU, i.e.,

\begin{equation}
\left( {{\rm{P1}}{\rm{.1}}} \right)\begin{array}{*{20}{l}}
{\mathop {\max }\limits_{\bf{u}} }&{\bar\gamma _R^{\left( \text{I} \right)} = {\bf{u}}_{}^H{\bf{\Phi }}_1\left( {\bf{u}} \right){\bf{u}}}\\[3mm]
{{\rm{s}}{\rm{.t}}{\rm{.}}}&{\bar\gamma _C= {\left| {{{\bf{h}}^H}{\bf{u}}} \right|^2} \Gamma }\\[1.5mm]
{}&{{\bf{u}}_{}^H{\bf{u}} \le {P_0}}
\end{array},
\end{equation}
where $\Gamma$ is the threshold of the CU’s SINR, and $P_0$ is the transmit power constraint of the DFRC BS. 

Because $\bar\gamma _R^{\left( \text{I} \right)}$ is a nonlinear function of the transmit beamforming vector $\bf{u}$, Problem (P1.1) is generally non-convex. Thus, we adopt the sequential optimization to find the transmit beamforming vector ${\bf{u}}$  in an iterative fashion. Specifically, at the $m$-th iteration, we first compute ${{\bf{\Phi }}_0} = {\bf{\Phi }}_1\left[ {{\bf{u}}_{}^{\left( {m - 1} \right)}} \right]$ , where ${\bf{u}}_{}^{\left( {m - 1} \right)}$  is obtained in the $(m-1)$-th iteration. Thus, Problem (P1.1) can be rewritten as 

\begin{equation}
\left( {{\rm{P1}}{\rm{.2}}} \right)\begin{array}{*{20}{l}}
{\mathop {\max }\limits_{\bf{u}} }&{\bar\gamma _R^{\left( \text{I} \right)} = {\bf{u}}_{}^H{\bf{\Phi }_0}{\bf{u}}}\\[3mm]
{{\rm{s}}{\rm{.t}}{\rm{.}}}&{\bar\gamma _C^{}= {\left| {{{\bf{h}}^H}{\bf{u}}} \right|^2} \ge \Gamma }\\[1.5mm]
{}&{{\bf{u}}_{}^H{\bf{u}} \le {P_0}}
\end{array}.
\end{equation}
And the closed-form solution of Problem (P1.2) can be given by the following proposition. 

\begin{proposition} (Optimal beamforming with single CU) 
For the non-dedicated probing signal case,  the optimal beamforming vector of the DFRC BS with single CU i.e., the optimal solution to Problem (P1.2), can be given by

\begin{equation}
\label{eq15}
{\bf{u}}_{}^* = \left\{ {\begin{array}{*{20}{l}}
{\sqrt {{P_0}} {\bf{\hat g}},}&{\Gamma  \le {P_0}{{\left| {{{\bf{h}}^H}{\bf{\hat g}}} \right|}^2}}\\[1.5mm]
{\left( {\alpha {\bf{\hat h}} + \beta {{{\bf{\hat g}}}_ \bot }} \right),}&{{P_0}{\left\| {\bf{h}} \right\|^2} \ge \Gamma  > {P_0}{{\left| {{{\bf{h}}^H}{\bf{\hat g}}} \right|}^2}}
\end{array}} \right.,
\end{equation}
\begin{equation}
\label{eq16}
\alpha  = \sqrt {\frac{\Gamma }{{{\left\| {\bf{h}} \right\|^2}}}} \frac{{{\alpha _g}}}{{\left| {{\alpha _g}} \right|}},
\end{equation}
\begin{equation}
\label{eq17}
\beta  = \sqrt {{P_0} - \frac{\Gamma }{{{\left\| {\bf{h}} \right\|^2}}}} \frac{{{\beta _g}}}{{\left| {{\beta _g}} \right|}},
\end{equation}
where ${\bf{g}}$ is the dominant eigenvector of  ${{\bf{\Phi }}_0}$, ${\bf{\hat g}} = {{\bf{g}} \mathord{\left/
 {\vphantom {{\bf{g}} \left\| {\bf{g}} \right\|}} \right.
 \kern-\nulldelimiterspace} \left\| {\bf{g}} \right\|}$, ${\bf{\hat h}} = {{{{\bf{h}}}} \mathord{\left/
 {\vphantom {{{{\bf{h}}}} {\left\| {\bf{h}} \right\|}}} \right.
 \kern-\nulldelimiterspace} {\left\| {\bf{h}} \right\|}}$, 
${{\bf{g}}_ \bot } = {\bf{g}} - ( {{{{\bf{\hat h}}}^H}{\bf{g}}} ){\bf{\hat h}}$ denoting the projection of ${\bf{g}}$  into the null space of  ${\bf{\hat h}}$, ${{\bf{\hat g}}_ \bot } = {{{\bf{g}}_ \bot } \mathord{\left/
 {\vphantom {{{\bf{g}}_ \bot } {\left\| {{\bf{g}}_ \bot ^{}} \right\|}}} \right.
 \kern-\nulldelimiterspace} {\left\| {{\bf{g}}_ \bot ^{}} \right\|}}$ and ${\bf{ g}}$  can be expressed as ${\bf{ g}} = {\alpha _g}{\bf{\hat h}} + {\beta _g}{{\bf{\hat g}}_ \bot }$.

\end{proposition}
\begin{proof}
The proof is given in Appendix \ref{proof1}.
\end{proof}

\subsection{Dedicated Probing Signal Case}
For the dedicated probing signal case, the transmit signal of the DFRC BS in (\ref{eqa5}) with single CU can be rewritten as

\begin{equation}
{\bf{x}} = {\bf{u}}s + {\bf{v}}{s_0},
\end{equation}
where ${\bf{v}} \in {{\cal C}^{{N_t} \times 1}}$ and ${s_0} \sim {\cal C}{\cal N}\left( {0,1} \right)$ are the beamforming vector and the symbol of the dedicated probing signal, respectively. In addition, $s$ and $s_0$ are independent and identically distributed (i.i.d.).

According to the output in (\ref{eq2}), the SINR of the radar receiver can be expressed as

\begin{equation}
\label{eq20}
\begin{aligned}
\gamma _R^{\left( \text{II} \right)} &= \frac{{{{\left| {{\alpha _0}{{\bf{w}}^H}{\bf{A}}\left( {{\theta _0}} \right){\bf{x}}} \right|}^2}}}{{{\rm{E}}\left[ {{{\left| {{{\bf{w}}^H}\sum\limits_{i = 1}^I {{\alpha _i}{\bf{A}}\left( {{\theta _i}} \right){\bf{x}}} } \right|}^2}} \right] + {{\bf{w}}^H}{\bf{w}}}} \\
&= \frac{{{{\left| {{\alpha _0}} \right|}^2}{{\left| {{{\bf{w}}^H}{\bf{A}}\left( {{\theta _0}} \right){\bf{x}}} \right|}^2}}}{{{{\bf{w}}^H}\left[ {\sum\limits_{i = 1}^I {{{\left| {{\alpha _i}} \right|}^2}{\bf{A}}\left( {{\theta _i}} \right)\left( {{\bf{uu}}_{}^H{\rm{ + }}{\bf{v}}{{\bf{v}}^H}} \right){{\bf{A}}^H}\left( {{\theta _i}} \right)}  + {\bf{I}}} \right]{\bf{w}}}}.
\end{aligned}
\end{equation}
By solving an equivalent MVDR problem, the corresponding receive beamforming vector to maximize the output SINR can be given by

\begin{equation}
\label{eq21}
{\bf{w}}^* = \frac{{{\bf{\Sigma }}_2{{\left( {{\bf{u}},{\bf{v}}} \right)}^{ - 1}}{\bf{A}}\left( {{\theta _0}} \right){\bf{x}}}}{{{{\bf{x}}^H}{{\bf{A}}^H}\left( {{\theta _0}} \right){\bf{\Sigma }}_2{{\left( {{\bf{u}},{\bf{v}}} \right)}^{ - 1}}{\bf{A}}\left( {{\theta _0}} \right){\bf{x}}}},
\end{equation}
where 

\begin{equation}
{\bf{\Sigma }}_2\left( {{\bf{u}},{\bf{v}}} \right) = \left[ {\sum\limits_{i = 1}^I {{{\left| {{\alpha _i}} \right|}^2}{\bf{A}}\left( {{\theta _i}} \right)\left( {{\bf{uu}}_{}^H{\rm{ + }}{\bf{v}}{{\bf{v}}^H}} \right){{\bf{A}}^H}\left( {{\theta _i}} \right)}  + {\bf{I}}} \right].
\end{equation}

Substituting (\ref{eq21}) into (\ref{eq20}), the SINR of the radar receiver can be calculated as

\begin{equation}
\gamma _R^{\left( \text{II} \right)} = {{\bf{x}}^H}{\bf{\Phi }}_2\left( {{\bf{u}},{\bf{v}}} \right){\bf{x}},
\end{equation}
where

\begin{equation}
{\bf{\Phi }}_2\left( {{\bf{u}},{\bf{v}}} \right) = {\left| {{\alpha _0}} \right|^2}{{\bf{A}}^H}\left( {{\theta _0}} \right){\bf{\Sigma }}_2{\left( {{\bf{u}},{\bf{v}}} \right)^{ - 1}}{\bf{A}}\left( {{\theta _0}} \right).
\end{equation}
And the average SINR of the radar receiver can be given by

\begin{equation}
\begin{aligned}
\bar\gamma _R^{\left( \text{II} \right)}\left( {{\bf{u}},{\bf{v}}} \right) &= {\rm{E}}\left[ {{{\bf{x}}^H}{\bf{\Phi }}_2\left( {{\bf{u}},{\bf{v}}} \right){\bf{x}}} \right]\\
&={\bf{u}}_{}^H{\bf{\Phi }}_2\left( {{\bf{u}},{\bf{v}}} \right){\bf{u}} + {{\bf{v}}^H}{\bf{\Phi }}_2\left( {{\bf{u}},{\bf{v}}} \right){\bf{v}}
\end{aligned}.
\end{equation}

For the CU, it has \emph{a priori} information of probing signal. After probing signal interference cancelling, its received SINR $\bar\gamma _C^{}\left( {\bf{u}} \right)$ can  also be expressed as (\ref{eq12}). 
The  beamforming  optimization  problem that  maximizes  the  SINR  of  the  radar  receiver  and  ensures the SINR requirement of the CU can be expressed as

\begin{equation}
\left( {{\rm{P2}}.1} \right)\begin{array}{*{20}{l}}
{\mathop {\max }\limits_{\bf{u}} }&{\bar \gamma _R^{\left( {{\rm{II}}} \right)} = {\bf{u}}_{}^H{{\bf{\Phi }}_2}\left( {{\bf{u}},{\bf{v}}} \right){\bf{u}} + {{\bf{v}}^H}{{\bf{\Phi }}_2}\left( {{\bf{u}},{\bf{v}}} \right){\bf{v}}}\\[3mm]
{{\rm{s}}.{\rm{t}}.}&{\bar \gamma _C^{} = {{\left| {{{\bf{h}}^H}{\bf{u}}} \right|}^2} \ge \Gamma }\\[1.5mm]
{}&{{{\bf{u}}^H}{\bf{u}} + {{\bf{v}}^H}{\bf{v}} \le {P_0}}
\end{array},    
\end{equation}

Similarly,  the sequential optimization can be adopted to find the transmit beamforming vector ${\bf{u}}$  and ${\bf{v}}$ in an iterative fashion, where we first compute ${{\bf{\Phi }}_0} = {\bf{\Phi }}_2\left[ {{\bf{u}}_{}^{\left( {m - 1} \right)}{\bf{v}}_{}^{\left( {m - 1} \right)}} \right]$  at the $m$-th iteration with  ${\bf{u}}_{}^{\left( {m - 1} \right)}$ and ${\bf{v}}_{}^{\left( {m - 1} \right)}$ being obtained from the $(m-1)$-th iteration. Thus, the beamforming optimization problem can be given by

\begin{equation}
\left( {{\rm{P2}}{\rm{.2}}} \right)\begin{array}{*{20}{l}}
{\mathop {\max }\limits_{\bf{u}} }&{\bar\gamma _R^{\left( \text{II} \right)} = {\bf{u}}_{}^H{{\bf{\Phi }}_0}{\bf{u}} + {{\bf{v}}^H}{{\bf{\Phi }}_0}{\bf{v}}}\\[3mm]
{{\rm{s}}{\rm{.t}}{\rm{.}}}&{\bar \gamma _C^{} = {\left| {{{\bf{h}}^H}{\bf{u}}} \right|^2} \ge \Gamma  }\\[1.5mm]
{}&{{{\bf{u}}^H}{\bf{u}} + {{\bf{v}}^H}{\bf{v}} \le {P_0}}
\end{array},
\end{equation}
And the closed-form solution of Problem (P2.2) can be given by the following proposition. 

\begin{proposition} (No need of dedicated probing signal with single CU)  For the dedicated probing signal case, the optimal beamforming vector of the dedicated probing signal with single CU, i.e., the optimal solution to Problem (P2.2), can be given by
${{\bf{v}}^*} = {\bf{0}}$. Thus, there is no need to design the dedicated probing signal with single CU, and the optimal beamforming vector of the communication signal ${{\bf{u}}^*}$ is also given by (\ref{eq15}). 
\end{proposition}
\begin{proof}
The proof is given in Appendix \ref{proof2}.
\end{proof}

Finally, the beamforming optimization algorithm for the single CU scencario can be summarized as Algorithm \ref{alg1}, which repeatedly updates  ${\bf{u}}_{}^{\left( {m-1} \right)}$ based on ${\bf{u}}_{}^{\left( m \right)}$ until the improvement of the radar receiver’s SINR becomes insignificant. 

\begin{algorithm}[t]
\caption{Beamforming design of DFRC BS for the single CU scenario.} 
\label{alg1}  
\begin{algorithmic}  
\STATE Initialize $\left\{ {{\theta _0},{\theta _1}, \cdots ,{\theta _I}} \right\}$ and $\left\{ {{\alpha _0},{\alpha _1}, \cdots ,{\alpha _I}} \right\}$.
\STATE Initialize ${{\bf{u}}^{\left( 0 \right)}} = {{\left[ {1, \cdots ,1} \right]^T} \mathord{\left/
 {\vphantom {{\left[ {1, \cdots ,1} \right]} {\sqrt {{P_0}{N_t}} }}} \right.
 \kern-\nulldelimiterspace} {\sqrt {{P_0}{N_t}} }}$.
\STATE Initialize the convergence threshold $\Delta$, and $m=0$.
\REPEAT
\STATE Set $m = m + 1$.
\STATE Calculate ${{\bf{\Phi }}_0} = {\bf{\Phi }}_1\left[ {{\bf{u}}_{}^{\left( {m - 1} \right)}} \right]$ and $\bar \gamma _R^{\left( \text{I} \right)}\left[ {{\bf{u}}_{}^{\left( {m - 1} \right)}} \right]$ according to (\ref{eq10}) and (\ref{eq11}), respectively. 
\STATE Calculate ${\bf{u}}_{}^{\left( m \right)}$ according to (\ref{eq15}) in Proposition 1.
\STATE Calculate $\bar \gamma _R^{\left( \text{I} \right)}\left[ {{\bf{u}}_{}^{\left( m \right)}} \right]$ according to (\ref{eq11}).
\UNTIL{$\left| {\bar\gamma _R^{\left( \text{I} \right)}\left[ {{\bf{u}}_{}^{\left( m \right]}} \right) - \bar\gamma _R^{\left( \text{I} \right)}\left[ {{\bf{u}}_{}^{\left( {m - 1} \right)}} \right]} \right| \le \Delta $}.
\end{algorithmic}  
\end{algorithm}

\section{Multiple CUs Scenario}
In this section, we proceed to consider a more complicated scenario with multiple CUs in the network. For the non-dedicated probing signal case, the beamforming design is formulated as a non-convex QCQP, and the globally optimal solutions can be obtained by applying SDR with rank-1 property. For the dedicated probing signal case, the rank-1 property after applying SDR can also be  proved, and the corresponding optimal solution shows that the dedicated probing signal should be employed to improve the SINR of the radar receiver.

\subsection{Non-dedicated Probing Signal Case}
For the non-dedicated probing signal case, the transmit signal of the DFRC BS in (\ref{eqa4}) with multiple CUs can be given by

\begin{equation}
{\bf{x}} = \sum\limits_{k = 1}^K {{{\bf{u}}_k}{s_k}},
\end{equation}
where ${{\bf{u}}_k} \in {{\cal C}^{{N_t} \times 1}}$ and ${s_k} \in {\cal C}{\cal N}\left( {0,1} \right)$ are the beamforming vector and the i.i.d. information symbol of the CU $k$, respectively.

According to the output in (\ref{eq2}), the SINR of the radar receiver can be written as

\begin{equation}
\label{eq36}
\begin{aligned}
\gamma _R^{\left( \text{I} \right)} &= \frac{{{{\left| {{\alpha _0}{{\bf{w}}^H}{\bf{A}}\left( {{\theta _0}} \right){\bf{x}}} \right|}^2}}}{{{\rm{E}}\left[ {{{\left| {{{\bf{w}}^H}\sum\limits_{i = 1}^I {{\alpha _i}{\bf{A}}\left( {{\theta _i}} \right){\bf{x}}} } \right|}^2}} \right] + {{\bf{w}}^H}{\bf{w}}}} \\
&= \frac{{{{\left| {{\alpha _0}} \right|}^2}{{\left| {{{\bf{w}}^H}{\bf{A}}\left( {{\theta _0}} \right){\bf{x}}} \right|}^2}}}{{{{\bf{w}}^H}\left[ {\sum\limits_{i = 1}^I {{{\left| {{\alpha _i}} \right|}^2}{\bf{A}}\left( {{\theta _i}} \right)\left( {\sum\limits_{k = 1}^K {{{\bf{u}}_k}{\bf{u}}_k^H} } \right){{\bf{A}}^H}\left( {{\theta _i}} \right)}  + {\bf{I}}} \right]{\bf{w}}}}.
\end{aligned}
\end{equation}

By solving an equivalent MVDR problem, the corresponding receive beamforming vector to maximize the output SINR can be given by

\begin{equation}
\label{eq37}
{{\bf{w}}^*} = \frac{{{{\bf{\Sigma }}_3}{{\left( {\left\{ {{{\bf{u}}_k}} \right\}} \right)}^{ - 1}}{\bf{A}}\left( {{\theta _0}} \right){\bf{x}}}}{{{{\bf{x}}^H}{{\bf{A}}^H}\left( {{\theta _0}} \right){{\bf{\Sigma }}_3}{{\left( {\left\{ {{{\bf{u}}_k}} \right\}} \right)}^{ - 1}}{\bf{A}}\left( {{\theta _0}} \right){\bf{x}}}},
\end{equation}
where 

\begin{equation}
{\bf{\Sigma }}_3\left( {\left\{ {{{\bf{u}}_k}} \right\}} \right) = \left[ {\sum\limits_{i = 1}^I {{{\left| {{\alpha _i}} \right|}^2}{\bf{A}}\left( {{\theta _i}} \right)} \left( {\sum\limits_{k = 1}^K {{{\bf{u}}_k}{\bf{u}}_k^H} } \right){{\bf{A}}^H}\left( {{\theta _i}} \right) + {\bf{I}}} \right].    
\end{equation}
Substituting (\ref{eq37}) into (\ref{eq36}), the SINR of the radar receiver can be calculated as

\begin{equation}
\gamma _R^{\left( \text{I} \right)}{\rm{ = }}{{\bf{x}}^H}{\bf{\Phi }}_3\left( {\left\{ {{{\bf{u}}_k}} \right\}} \right){\bf{x}},    
\end{equation}
where 

\begin{equation}
\label{eq40}
{\bf{\Phi }}_3\left( {\left\{ {{{\bf{u}}_k}} \right\}} \right) = {\left| {{\alpha _0}} \right|^2}{{\bf{A}}^H}\left( {{\theta _0}} \right){\bf{\Sigma }}_3{\left( {\left\{ {{{\bf{u}}_k}} \right\}} \right)^{ - 1}}{\bf{A}}\left( {{\theta _0}} \right).
\end{equation}
And the average SINR  of the radar receiver can be given by

\begin{equation}
\label{eq41}
\begin{aligned}
\bar\gamma _R^{\left( \text{I} \right)} &= {\rm{E}}\left[ {{{\bf{x}}^H}{\bf{\Phi }}_3\left( {\left\{ {{{\bf{u}}_k}} \right\}} \right){\bf{x}}} \right] \\
&= \sum\limits_{k = 1}^K {{\bf{u}}_k^H{\bf{\Phi }}_3\left( {\left\{ {{{\bf{u}}_k}} \right\}} \right)} {{\bf{u}}_k}
\end{aligned}
\end{equation}
The received signal of the CU $k$ in (\ref{eq3}) can be rewritten as 

\begin{equation}
{y_k} = {\bf{h}}_k^H{\bf{u}}_k^{}{s_k} + \sum\limits_{j \ne k}^{} {{\bf{h}}_k^H{\bf{u}}_j^{}{s_j}}  + {z_k},
\end{equation}
and the average SINR of the CU $k$ can be calculated as

\begin{equation}
\label{eq43}
\bar \gamma _{C,k}^{}\left( {\left\{ {{{\bf{u}}_k}} \right\}} \right) = \frac{{{{\left| {{\bf{h}}_k^H{\bf{u}}_k^{}} \right|}^2}}}{{1 + \sum\limits_{j \ne k}^{} {{{\left| {{\bf{h}}_k^H{\bf{u}}_j^{}} \right|}^2}} }}. 
\end{equation}

Again, we consider the beamforming optimization problem that maximizes the SINR of the radar receiver by imposing the individual SINR constraints of the CUs, i.e.,

\begin{equation}
\left( {{\rm{P3}}.{\rm{1}}} \right)\begin{array}{*{20}{l}}
{\mathop {\max }\limits_{\left\{ {{{\bf{u}}_k}} \right\}} }&{\bar \gamma _R^{\left( {\rm{I}} \right)} = \sum\limits_{k = 1}^K {{\bf{u}}_k^H{{\bf{\Phi }}_3}\left( {\left\{ {{{\bf{u}}_{\bf{k}}}} \right\}} \right)} {{\bf{u}}_k}}\\[3mm]
{{\rm{s}}.{\rm{t}}.}&{\bar \gamma _{C,k}^{} = \dfrac{{{{\left| {{\bf{h}}_k^H{\bf{u}}_k^{}} \right|}^2}}}{{1 + \sum\limits_{j \ne k}^{} {{{\left| {{\bf{h}}_k^H{\bf{u}}_j^{}} \right|}^2}} }} \ge {\Gamma _k},\forall k}\\[1.5mm]
{}&{\sum\limits_{k = 1}^K {{\bf{u}}_k^H} {{\bf{u}}_k} \le {P_0}}
\end{array}.
\end{equation}
Note that Problem (P3.1) is also non-convex. Thus, we adopt the sequential optimization to find the transmit beamforming vector ${\bf{u}}$  in an iterative way. Specifically, at the $m$-th iteration, we first compute ${{\bf{\Phi }}_0} = {\bf{\Phi }}_3[ {\{ {{\bf{u}}_k^{\left( {m - 1} \right)}} \}} ]$ , where ${\{ {{\bf{u}}_k^{\left( {m - 1} \right)}} \}}$  is obtained in the $(m-1)$-th iteration. As a result, Problem (P3.1) can be rewritten as

\begin{equation}
\left( {{\rm{P3}}.{\rm{2}}} \right)\begin{array}{*{20}{l}}
{\mathop {\max }\limits_{\left\{ {{{\bf{u}}_k}} \right\}} }&{\bar \gamma _R^{\left( {\rm{I}} \right)} = \sum\limits_{k = 1}^K {{\bf{u}}_k^H{{\bf{\Phi }}_0}} {{\bf{u}}_k}}\\[3mm]
{{\rm{s}}.{\rm{t}}.}&{\bar \gamma _{C,k}^{} = \dfrac{{{{\left| {{\bf{h}}_k^H{\bf{u}}_k^{}} \right|}^2}}}{{1 + \sum\limits_{j \ne k}^{} {{{\left| {{\bf{h}}_k^H{\bf{u}}_j^{}} \right|}^2}} }} \ge {\Gamma _k},\forall k}\\[1.5mm]
{}&{\sum\limits_{k = 1}^K {{\bf{u}}_k^H} {{\bf{u}}_k} \le {P_0}}
\end{array}
\end{equation}
While Problem (P3.2) is still a non-convex QCQP, and we can solve it via SDR with ${{\bf{U}}_k} = {{\bf{u}}_k}{\bf{u}}_k^H$, i.e,

\begin{equation}
\left( {{\rm{P3}}{\rm{.3}}} \right)\begin{array}{*{20}{l}}
{\mathop {\max }\limits_{\left\{ {{{\bf{U}}_k}} \right\}} }&{\bar \gamma _R^{\left( \text{I} \right)} = \sum\limits_{k = 1}^K {{\rm{tr}}\left( {{{\bf{\Phi }}_0}{{\bf{U}}_k}} \right)} }\\[3mm]
{{\rm{s}}{\rm{.t}}{\rm{.}}}&{\bar \gamma _{C,k}^{} = \dfrac{{{\rm{tr}}\left( {{{\bf{H}}_k}{{\bf{U}}_k}} \right)}}{{{\Gamma _k}}} - \sum\limits_{j \ne k} {{\rm{tr}}\left( {{{\bf{H}}_k}{{\bf{U}}_j}} \right)}  \ge 1,\forall k}\\[1.5mm]
{}&{\sum\limits_{k = 1}^K {{\rm{tr}}\left( {{{\bf{U}}_k}} \right)}  \le {P_0},{{\bf{U}}_k} \succeq 0,\forall k}
\end{array},
\end{equation}
where ${{\bf{H}}_k} = {\bf{h}}_k^{}{\bf{h}}_k^H$.

Although the rank-1 constraints have been removed for the convexity of the problem, the optimal solution can be proved to have the rank-1 property by the following proposition. 

\begin{proposition} (Rank-1 property of Problem (P3.3)) 
For the non-dedicated probing signal case with multiple CUs, 
there is always a solution to Problem (P3.3) satisfying that ${\rm{rank}}\left( {{\bf{U}}_k^*} \right) = 1,\forall k$. Thus, the optimal beamforming vector for the CU $k$, i.e., the optimal solution to Problem (P3.2) can be given by ${\bf{u}}_k^*$ with ${\bf{u}}_k^*{\bf{u}}_k^{*H} = {\bf{U}}_k^*$.
\end{proposition}
\begin{proof}
The proof is given in Appendix \ref{proof3}.
\end{proof}

\begin{algorithm}[t]
\caption{Beamforming design of DFRC BS for multiple CUs  scenario without dedicated probing signal.} 
\label{alg2}  
\begin{algorithmic}  
\STATE Initialize $\left\{ {{\theta _0},{\theta _1}, \cdots ,{\theta _I}} \right\}$ and $\left\{ {{\alpha _0},{\alpha _1}, \cdots ,{\alpha _I}} \right\}$.
\STATE Initialize $\left\{ {{{\bf{u}}_k^{(0)}}} \right\} = {{{{\left[ {1, \cdots ,1} \right]}^T}} \mathord{\left/
 {\vphantom {{{{\left[ {1, \cdots ,1} \right]}^T}} {\sqrt {K{P_0}{N_t}} }}} \right.
 \kern-\nulldelimiterspace} {\sqrt {K{P_0}{N_t}} }}$.
\STATE Initialize the convergence threshold $\Delta$, and $m=0$.
\REPEAT
\STATE Set $m = m + 1$.
\STATE Calculate ${{\bf{\Phi }}_0} = {\bf{\Phi }}_3\left[ {\left\{ {{\bf{u}}_k^{\left( {m - 1} \right)}} \right\}} \right]$ and $\bar\gamma _R^{\left( \text{I} \right)}\left[ {\left\{ {{\bf{u}}_k^{\left( {m - 1} \right)}} \right\}} \right]$ according to (\ref{eq40}) and (\ref{eq41}), respectively.
\STATE Optimize $\left\{ {{\bf{U}}_k^*} \right\}$  according to  Problem (P3.3).
\STATE Calculate $\{{\bf{u}}_k^*\}$ satisfying ${\bf{u}}_k^*{\left( {{\bf{u}}_k^*} \right)^H} = {\bf{U}}_k^*$ based on the Proposition 3.
\STATE Calculate $\bar\gamma _R^{\left( \text{I} \right)}\left[ {\left\{ {{\bf{u}}_k^{\left( m \right)}} \right\}} \right]$ according to (\ref{eq41}).
\UNTIL{$\left| {\bar\gamma _R^{\left( \text{I} \right)}\left[ {\left\{ {{\bf{u}}_k^{\left( m \right)}} \right\}} \right] - \bar\gamma _R^{\left( \text{I} \right)}\left[ {\left\{ {{\bf{u}}_k^{\left( {m{\rm{ - }}1} \right)}} \right\}} \right]} \right| \le \Delta $}.
\end{algorithmic}  
\end{algorithm}

Above all, the beamforming optimization algorithm for the multi-CU scenario and the non-dedicated probing signal case can be provided as Algorithm \ref{alg2}.
It repeatedly updates $\{ {{\bf{u}}_k^{\left( {m} \right)}} \}$ given $\{ {{\bf{u}}_k^{\left( {m-1} \right)}} \}$ until convergence.

\subsection{Dedicated Probing Signal Case}
For the dedicated probing signal case, the transmit signal of the DFRC BS in (\ref{eqa5}) with multiple CUs can be given by

\begin{equation}
{\bf{x}} = \sum\limits_{k=1}^{K} {{{\bf{u}}_k}{s_k}}  + {\bf{v}}{s_0},
\end{equation}
where ${\bf{v}} \in {{\cal C}^{{N_t} \times 1}}$ and ${s_0} \sim {\cal C}{\cal N}\left( {0,1} \right)$ are the beamforming vector and the symbol of the dedicated probing signal, respectively. And $s_k, \forall k$ and $s_0$ are i.i.d.. 

According to the output in (\ref{eq2}), the SINR of the radar receiver can be expressed as

\begin{equation}
\label{eq63}
\begin{aligned}
&\gamma _R^{\left( \text{II} \right)} = \frac{{{{\left| {{\alpha _0}{{\bf{w}}^H}{\bf{A}}\left( {{\theta _0}} \right){\bf{x}}} \right|}^2}}}{{{\rm{E}}\left[ {{{\left| {{{\bf{w}}^H}\sum\limits_{i = 1}^I {{\alpha _i}{\bf{A}}\left( {{\theta _i}} \right){\bf{x}}} } \right|}^2}} \right] + {{\bf{w}}^H}{\bf{w}}}} \\
&= \frac{{{{\left| {{\alpha _0}} \right|}^2}{{\left| {{{\bf{w}}^H}{\bf{A}}\left( {{\theta _0}} \right){\bf{x}}} \right|}^2}}}{{{{\bf{w}}^H}\left[ {\sum\limits_{i = 1}^I {{{\left| {{\alpha _i}} \right|}^2}{\bf{A}}\left( {{\theta _i}} \right)\left( {\sum\limits_{k = 1}^K {{{\bf{u}}_k}{\bf{u}}_k^H} \! +\!  {\bf{v}}{{\bf{v}}^H}} \right){{\bf{A}}^H}\left( {{\theta _i}} \right)} \!  +\!  {\bf{I}}} \right]{\bf{w}}}}    
\end{aligned}
\end{equation}
By solving an equivalent MVDR problem, the corresponding receive beamforming vector to maximize the output SINR can be given by

\begin{equation}
\label{eq64}
{\bf{w}}^* = \frac{{{\bf{\Sigma }}_4{{\left( {\left\{ {{{\bf{u}}_k}} \right\},{\bf{v}}} \right)}^{ - 1}}{\bf{A}}\left( {{\theta _0}} \right){\bf{x}}}}{{{{\bf{x}}^H}{{\bf{A}}^H}\left( {{\theta _0}} \right){\bf{\Sigma }}_4{{\left( {\left\{ {{{\bf{u}}_k}} \right\},{\bf{v}}} \right)}^{ - 1}}{\bf{A}}\left( {{\theta _0}} \right){\bf{x}}}},
\end{equation}
where 

\begin{equation}
{\bf{\Sigma }}_4\left( {\left\{ {{{\bf{u}}_k}} \right\},{\bf{v}}} \right)\!  \! =\! \!  {\sum\limits_{i = 1}^I {{{\left| {{\alpha _i}} \right|}^2}{\bf{A}}\left( {{\theta _i}} \right)}\!  \left( {\sum\limits_{k = 1}^K {{{\bf{u}}_k}{\bf{u}}_k^H} \! \!  + \! \! {\bf{v}}{{\bf{v}}^H}} \right)\! {{\bf{A}}^H}\left( {{\theta _i}} \right)\! \!  +\! \!  {\bf{I}}}.
\end{equation}
Substituting (\ref{eq64}) into (\ref{eq63}), the SINR of the radar receiver can be calculated as

\begin{equation}
\gamma _R^{\left( \text{II} \right)} = {{\bf{x}}^H}{\bf{\Phi }}_4\left( {\left\{ {{{\bf{u}}_k}} \right\},{\bf{v}}} \right){\bf{x}},
\end{equation}
where

\begin{equation}
\label{eqa70}
{\bf{\Phi }}_4\left( {\left\{ {{{\bf{u}}_k}} \right\},{\bf{v}}} \right) = {\left| {{\alpha _0}} \right|^2}{{\bf{A}}^H}\left( {{\theta _0}} \right){\bf{\Sigma }}_4{\left( {\left\{ {{{\bf{u}}_k}} \right\},{\bf{v}}} \right)^{ - 1}}{\bf{A}}\left( {{\theta _0}} \right).
\end{equation}
And the average SINR of the radar receiver can be given by

\begin{equation}
\label{eqa71}
\begin{aligned}
\bar\gamma _R^{\left( \text{II} \right)} &= {\rm{E}}\left[ {{{\bf{x}}^H}{\bf{\Phi }}_4\left( {\left\{ {{{\bf{u}}_k}} \right\},{\bf{v}}} \right){\bf{x}}} \right]\\
&=\sum\limits_{k = 1}^K {{\bf{u}}_k^H{\bf{\Phi }}_4\left( {\left\{ {{{\bf{u}}_k}} \right\},{\bf{v}}} \right){{\bf{u}}_k}}  + {\bf{v}}_{}^H{\bf{\Phi }}_4\left( {\left\{ {{{\bf{u}}_k}} \right\},{\bf{v}}} \right){\bf{v}}
\end{aligned}
\end{equation}

For the CU $k$, it has \emph{a priori} information on the probing signal. After the probing signal interference cancelling, its received SINR $\bar \gamma _{C,k}^{}\left( {\left\{ {{{\bf{u}}_k}} \right\}} \right)$ can also be expressed as (\ref{eq43}). 
Thus, the beamforming optimization problem that maximizes the SINR of the radar receiver and satisfies the SINR constraints of the CUs can be formulated as

\begin{equation}
\left( {{\rm{P4}}.{\rm{1}}} \right)\begin{array}{*{20}{l}}
{\mathop {\max }\limits_{\left\{ {{{\bf{u}}_k}} \right\},{\bf{v}}} }&{\bar \gamma _R^{\left( {{\rm{II}}} \right)} = \sum\limits_{k = 1}^K {{\bf{u}}_k^H{{\bf{\Phi }}_4}{{\bf{u}}_k}}  + {\bf{v}}_{}^H{{\bf{\Phi }}_4}{\bf{v}}}\\[3mm]
{{\rm{s}}.{\rm{t}}.}&{\bar \gamma _{C,k}^{} = \dfrac{{{{\left| {{\bf{h}}_k^H{\bf{u}}_k^{}} \right|}^2}}}{{1 + \sum\limits_{j \ne k}^{} {{{\left| {{\bf{h}}_k^H{\bf{u}}_j^{}} \right|}^2}} }} \ge {\Gamma _k},\forall k}\\[1.5mm]
{}&{\sum\limits_{k = 1}^K {{\bf{u}}_k^H{{\bf{u}}_k}}  + {\bf{v}}_{}^H{\bf{v}} \le {P_0}}
\end{array}
\end{equation}

Let us then employ the sequential optimization to find the transmit beamforming vector ${\left\{ {{{\bf{u}}_k}} \right\}}$  and ${\bf{v}}$ in an iterative fashion, where we first compute ${{\bf{\Phi }}_0} = {\bf{\Phi }}_4[ {\{ {{\bf{u}}_k^{(m - 1)}} \},{\bf{v}}_{}^{(m - 1)}}]$  at the $m$-th iteration with  ${\{ {{\bf{u}}_k^{(m - 1)}} \}}$ and ${\bf{v}}_{}^{\left( {m - 1} \right)}$ being obtained from the $(m-1)$-th iteration. Thus, the beamforming  optimization problem can be again formulated as

\begin{equation}
\left( {{\rm{P4}}.{\rm{2}}} \right)\begin{array}{*{20}{l}}
{\mathop {\max }\limits_{\left\{ {{{\bf{u}}_k}} \right\},{\bf{v}}} }&{\bar \gamma _R^{\left( {{\rm{II}}} \right)} = \sum\limits_{k = 1}^K {{\bf{u}}_k^H{{\bf{\Phi }}_0}{{\bf{u}}_k}}  + {\bf{v}}_{}^H{{\bf{\Phi }}_0}{\bf{v}}}\\[3mm]
{{\rm{s}}.{\rm{t}}.}&{\bar \gamma _{C,k}^{} = \dfrac{{{{\left| {{\bf{h}}_k^H{\bf{u}}_k^{}} \right|}^2}}}{{1 + \sum\limits_{j \ne k}^{} {{{\left| {{\bf{h}}_k^H{\bf{u}}_j^{}} \right|}^2}} }} \ge {\Gamma _k},\forall k}\\[1.5mm]
{}&{\sum\limits_{k = 1}^K {{\bf{u}}_k^H{{\bf{u}}_k}}  + {\bf{v}}_{}^H{\bf{v}} \le {P_0}}
\end{array},
\end{equation}
which can be solved using the SDR by letting  ${{\bf{U}}_k} = {{\bf{u}}_k}{\bf{u}}_k^H$ and  ${{\bf{V}}_k} = {{\bf{v}}_k}{\bf{v}}_k^H$, i.e,

\begin{equation}
\left( {{\rm{P4}}{\rm{.3}}} \right)\begin{array}{*{20}{l}}
{\mathop {\max }\limits_{\left\{ {{{\bf{U}}_k}} \right\},{\bf{V}}} }&{\bar \gamma _R^{\left( \text{II} \right)} = \sum\limits_{k = 1}^K {{\rm{tr}}\left( {{{\bf{\Phi }}_0}{{\bf{U}}_k}} \right){\rm{ + tr}}\left( {{{\bf{\Phi }}_0}{\bf{V}}} \right)} }\\[3mm]
{{\rm{s}}{\rm{.t}}}&{\bar \gamma _{C,k}^{}\! =\! \dfrac{{{\rm{tr}}\left( {{{\bf{H}}_k}{{\bf{U}}_k}} \right)}}{{{\Gamma _k}}}\!\! - \!\!\sum\limits_{j \ne k} {{\rm{tr}}\left( {{{\bf{H}}_k}{{\bf{U}}_j}} \right)}  \ge 1,\forall k}\\[1.5mm]
{}&{\sum\limits_{k = 1}^K {{\rm{tr}}\left( {{{\bf{U}}_k}} \right)}  + {\rm{tr}}\left( {\bf{V}} \right) \le {P_0}}\\[1.5mm]
{}&{{{\bf{U}}_k} \succeq 0,\forall k,{\bf{V}} \succeq 0}
\end{array}.
\end{equation}

Although the rank-1 constraints have been removed for the convexity of the problem, the optimal solution can be guaranteed to have rank-1 property by the following proposition. Furthermore, the dedicated probing signal should be employed according to the following proposition. 

\begin{proposition} (Rank-1 property of Problem (P4.3)) 
For the non-dedicated probing signal case with multiple CUs, there is always a solution to Problem (P4.3) satisfying that ${\rm{rank}}\left( {{\bf{U}}_k^*} \right) = 1,\forall k$ and ${\rm{rank}}\left( {{{\bf{V}}^*}} \right) \le 1$. Thus, the optimal beamforming vector for the CU $k$ and the dedicated probing signal, i.e., the optimal solution to Problem (P4.2) can be given by ${\bf{u}}_k^*$  and ${\bf{v}}^*$ with ${\bf{u}}_k^*{\bf{u}}_k^{*H} = {\bf{U}}_k^*$ and ${\bf{v}}^*{\bf{v}}^{*H} = {\bf{V}}^*$. Specifically, ${{\bf{v}}^*} = \sqrt {\tau {P_0}} {\bf{\hat g}}$, where ${\bf{\hat g}} = {{\bf{g}} \mathord{\left/
 {\vphantom {{\bf{g}} \left\| {\bf{g}} \right\|}} \right.
 \kern-\nulldelimiterspace} \left\| {\bf{g}} \right\|}$, ${\bf{g}}$ is the dominant eigenvector of  ${{\bf{\Phi }}_0}$, and $0 \le \tau  \le 1$.
\end{proposition}
\begin{proof}
The proof is given in Appendix \ref{proof4}.
\end{proof}

\begin{remark} (Dedicated probing signal is employed or not)
It can be observed that any feasible solution to Problem (4.2) is also feasible for Problem (3.2) with  ${{\bf{v}}^*} = 0$, and vice versa. If  ${{\bf{v}}^*} \ne 0$, a higher SINR of radar receiver can be achievable, which will also be verified by the simulation results in the next section. Therefore, it is beneficial to employ the dedicated probing signal for multiple CUs scenario. The benefit is achieved at the cost of implementing an additional interference cancellation with \emph{a priori} known probing signals by all CUs. 
\end{remark}

Finally, we can also use the solution  $\{ {{\bf{u}}_k^{\left( {m-1} \right)}} \}$ and ${{\bf{v}}^{\left( {m-1} \right)}}$ to update $\{ {{\bf{u}}_k^{\left( {m} \right)}} \}$ and ${{\bf{v}}^{\left( {m} \right)}}$, and it is repeated until the improvement of the radar receiver’s SINR becomes insignificant as illustrated in Algorithm \ref{alg3}.

\begin{algorithm}[t]
\caption{Beamforming design of DFRC BS for multiple CUs  scenario with dedicated probing signal.} 
\label{alg3}  
\begin{algorithmic}  
\STATE Initialize $\left\{ {{\theta _0},{\theta _1}, \cdots ,{\theta _I}} \right\}$ and $\left\{ {{\alpha _0},{\alpha _1}, \cdots ,{\alpha _I}} \right\}$.
\STATE Initialize $\left\{ {{{\bf{u}}_k^{(0)}}} \right\} = {{\bf{v}}^{(0)}} = {{{{\left[ {1, \cdots ,1} \right]}^T}} \mathord{\left/
 {\vphantom {{{{\left[ {1, \cdots ,1} \right]}^T}} {\sqrt {\left( {K + 1} \right){P_0}{N_t}} }}} \right.
 \kern-\nulldelimiterspace} {\sqrt {\left( {K + 1} \right){P_0}{N_t}} }}.$
\STATE Initialize the convergence threshold $\Delta$, and $m=0$.
\REPEAT
\STATE Set $m = m + 1$.
\STATE Calculate ${{\bf{\Phi }}_0} = {{\bf{\Phi }}_4}\left[ {\left\{ {{\bf{u}}_k^{\left( {m - 1} \right)}} \right\},{\bf{v}}_{}^{\left( {m - 1} \right)}} \right]$ according to 
(\ref{eqa70})
and $\bar \gamma _R^{\left( {{\rm{II}}} \right)}\left[ {\left\{ {{\bf{u}}_k^{\left( {m - 1} \right)}} \right\},{\bf{v}}_{}^{\left( {m - 1} \right)}} \right]$ according to (\ref{eqa71}).
\STATE Optimize $\left\{ {{\bf{U}}_k^*} \right\}$ and ${{\bf{V}}^*}$  according to  Problem (P4.3);
\STATE Calculate $\{{\bf{u}}_k^*\}$ and ${\bf{v}}^*$ satisfying ${\bf{u}}_k^*{\left( {{\bf{u}}_k^*} \right)^H} = {\bf{U}}_k^*$ and ${{\bf{v}}^*}{{\bf{v}}^{*H}} = {{\bf{V}}^*}$
based on the Proposition 4.
\STATE Calculate $\bar \gamma _R^{\left( {{\rm{II}}} \right)}\left[ {\left\{ {{\bf{u}}_k^{\left( {m} \right)}} \right\},{\bf{v}}_{}^{\left( {m} \right)}} \right]$ according to (\ref{eqa71}).
\UNTIL{$\left| {\bar \gamma _R^{\left( {{\rm{II}}} \right)}\left[ {\left\{ {{\bf{u}}_k^{\left( m \right)}} \right\} , {\bf{v}}_{}^{\left( m \right)}} \right] - \bar \gamma _R^{\left( {{\rm{II}}} \right)}\left[ {\left\{ {{\bf{u}}_k^{\left( {m - 1} \right)}} \right\} , {\bf{v}}_{}^{\left( {m - 1} \right)}} \right]} \right| \le   \Delta$}.
\end{algorithmic}  
\end{algorithm}

\section{Simulation Results}
In this section, we evaluate the performance of our proposed design via numerical simulations. We assume that both the DFRC BS and the radar receiver are equipped with uniform linear arrays (ULAs) with the same number of elements. The interval between adjacent antennas of the DFRC  BS and the radar receiver is half-wavelength. 
The transmit power constraint of DFRC BS is set as $P_0=20$ dBm.
A target is located at the spatial angle $\theta_0=0^{\circ}$ with power ${{{\left| {{\alpha _0}} \right|}^2}}=10$ dB, and four fixed interferences are located at the spatial angels $\theta_1=-60^{\circ}$, $\theta_2=-30^{\circ}$, $\theta_3=30^{\circ}$, $\theta_4=60^{\circ}$, respectively. The power of each interference is ${{{\left| {{\alpha _i}} \right|}^2}}=30$ dB, $\forall i$. The channel vector of each CU is randomly generated from i.i.d. Rayleigh fading.

\subsection {Single CU Scenario}
In Fig. \ref{SUregion}, the tradeoff between the SINR constraint of CU and the SINR of radar is shown for the single CU scenario by varying the threshold of the CU's SINR $\Gamma$. It is easy to identify two boundary points of the tradeoff, i.e., ``Radar benchmark'' and ``Communication benchmark''. When the constraint of the CU's SINR is inactive, the radar's SINR is around $38$ dB and the CU's SINR is around $20$ dB. When the constraint of the CU's SINR is active, the maximum feasible constraint of the CU's SINR is around $28$ dB, and the corresponding radar's SINR is around $25$ dB. The optimal tradeoff between the SINR of CU and the SINR of radar is characterized by the solid line. It can be observed that the SINR of the radar decreases with the increase of the CU's SINR constraint, when the CU's SINR constraint is above $20$ dB.
Also, the tradeoff between the SINR of CU and the SINR of radar can be achieved by time sharing, which is shown by the dotted line. However, the optimal beamforming design for the DFRC BS yields better tradeoff performance.

\begin{figure}[t]
\centering
  \includegraphics[width=0.5\textwidth]{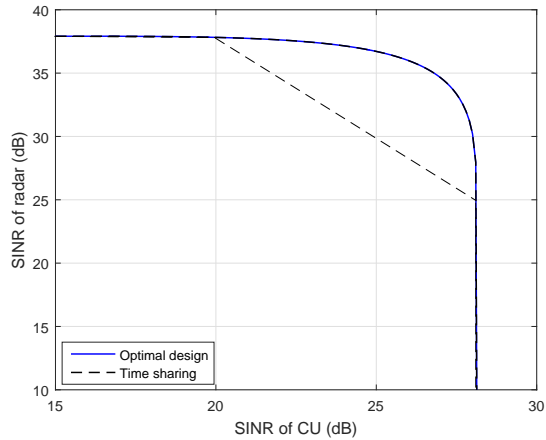}
\caption{Tradeoff between the SINR constraint of CU and the SINR of radar, $N_t=N_r=8$, ${\bf{h}}$ = [0.21 - 0.02i,-0.56 + 0.65i,0.57- 0.23i,-0.93 + 0.47i,-0.19 + 1.35i,-0.19 + 0.11i,1.05 - 0.21i,1.02 - 0.35i].}
\label{SUregion}
\end{figure}

In Fig. \ref{SUbeampattern}, the optimized beampatterns with different constraints of the CU's SINR are illustrated. 
The nulls are clearly placed at the locations of interferences, i.e., $\theta_1=-60^{\circ}$, $\theta_2=-30^{\circ}$, $\theta_3=30^{\circ}$, $\theta_4=60^{\circ}$, and the target is located at $\theta_0=0^{\circ}$.
When the SINR constraint of the CU is $15$ dB, it is inactive, and the beampattern is actually the optimal beampattern to maximize the radar's SINR. When the SINR constraint of the CU increases to $25$ dB, the performance of beampattern becomes worse from the radar's viewpoint. Particularly,  both the  peak  to  sidelobe  ratio (PSLR) and the main beam power both decrease with the increase of the SINR constraint of the CU. Furthermore, the optimized beampatterns with different numbers of the transmit and receive antennas are shown in Fig. \ref{Ntbeampattern}. The SINR constraint of the CU is fixed as $25$ dB. When the number of the DFRC BS antennas increases, 
the performance of beampattern becomes better from the radar's viewpoint. In particular,  the PSLR increases with the number of the DFRC BS antennas, and the main beam width decreases with the number of the DFRC BS antennas.

\begin{figure}[t]
\centering
  \includegraphics[width=0.5\textwidth]{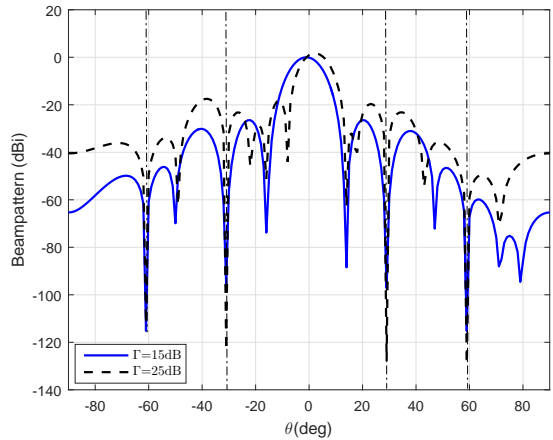}
\caption{Optimized beampatterns with different constraints of the CU's SINR, $N_t=N_r=8$, ${\bf{h}}= [0.21 - 0.02i,-0.56 + 0.65i,0.57- 0.23i,-0.93 + 0.47i,-0.19 + 1.35i,-0.19 + 0.11i,1.05 - 0.21i,1.02 - 0.35i]$ .}
\label{SUbeampattern}
\end{figure}

\begin{figure}[t]
\centering
  \includegraphics[width=0.5\textwidth]{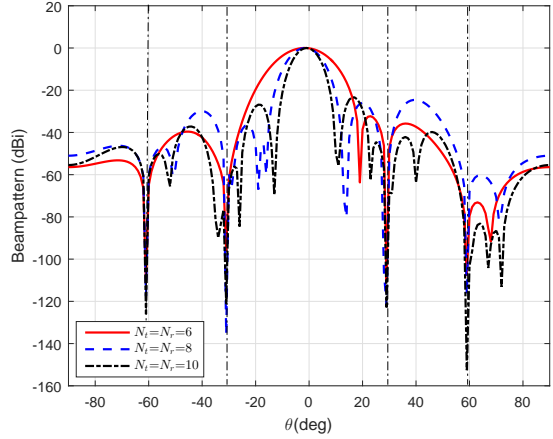}
\caption{Optimized beampatterns with different numbers of the DFRC BS antennas, $N_t=N_r$,
${\bf{h}}_1 = [1.43-0.47i,-0.83-0.56i,0.05-0.73i,-0.32 + 0.31i,0.21+0.68i,-0.56-0.95i]$, ${\bf{h}}_2$ = [${\bf{h}}_1, 0.23-0.62i,0.90+0.60i$], ${\bf{h}}_3$ = [${\bf{h}}_2,-0.22-0.83i,0.56-1.22i$], $\Gamma=20$ dB.}
\label{Ntbeampattern}
\end{figure}

\subsection {Multiple CUs Scenario}

In Fig. \ref{Iteration}, we evaluate the convergence performance of the proposed algorithm for different numbers of CUs. 
The SINR of the radar versus the number of iterations is provided. The SINR of the radar converges to a fixed value with only 2 iterations.
And the converged SINR performance decreases with the increase of the number of the CUs.
Furthermore, the same
converged SINR of the radar for dedicated probing signal case and non-dedicated probing signal case is observed, when the number of CU is 1. And the converged SINR of the radar improves with employing dedicated probing signal for the multiple CUs scenario.

\begin{figure}[t]
\centering
  \includegraphics[width=0.5\textwidth]{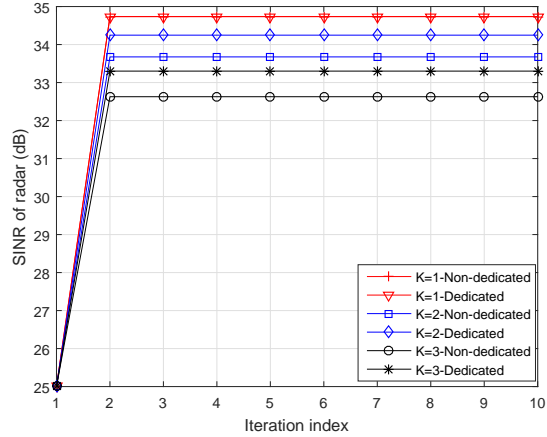}
\caption{Convergence performance of the proposed algorithm for different numbers of CU, $N_t=N_r=6$, $\Gamma=20$dB, ${\bf{h}}_1=[-0.33 - 1.17i,-0.47 - 0.61i,0.42 - 0.82i,0.22 - 0.71i,0.21 + 0.58i,-0.19 + 0.13i]$, ${\bf{h}}_2=[-0.69 + 0.67i,-0.38 + 0.79i,-0.32 + 1.09i,1.04 - 0.02i,-0.17 + 0.05i,-2.01 - 0.60i]$, ${\bf{h}}_3=[0.80 - 0.48i,-0.68 - 0.23i,0.10 - 0.17i,-0.01 - 0.19i,0.06 + 0.62i,-0.34 - 0.03i]$}
\label{Iteration}
\end{figure}

The average performance of the tradeoff between the SINR constraints of CUs and the average SINR of radar is evaluated in Fig. \ref{MUregion} through $10^4$ Monte Carlo simulations.
The SINR of radar exponentially decreases with the increase of the CUs' SINR constraint. And the average SINR of radar becomes worse when the number of the CU increases. That is because the feasible region of the radar beamforming optimization becomes smaller when either the  CUs' SINR constraint or the number of the CUs increase. Furthermore, it can be observed that the average SINR of radar is the same for the dedicated probing signal case and non-dedicated probing signal case, when the number of CU is 1. That verifies  Proposition 2. And, it can also observed that the average SINR of the radar improves with employing dedicated probing signal for multiple CUs scenario. In the case that the CUs' SINR constraints are loose, the improvement is not obvious because the CUs' SINR constraints are inactive. When the CUs' SINRs are constrained to be large, the improvement becomes obvious. 
For example, when the number of CU is 2 and CUs' SINR constraint is 10 dB, it improves about 1 dB.

\begin{figure}[t]
\centering
  \includegraphics[width=0.5\textwidth]{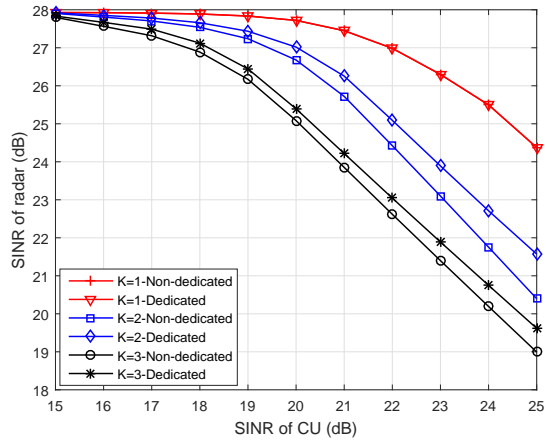}
\caption{Tradeoff between the SINR constraints of CUs and the average SINR of radar, $N_t=N_r=8$.}
\label{MUregion}
\end{figure}

In Fig. \ref{KSNR}, the average SINR of radar versus the numbers of CUs is illustrated with different number of transmit and receive antennas through $10^4$ Monte Carlo simulations.
The average SINR of radar decreases with the increase of the number of CUs, and it increases with the increase of the number of transmit and receive antennas. When the number of transmit and receive antennas is large, the decrease in the average SINR of radar with the number of CUs is not obvious. 
It can also be observed that the average SINR of radar improves by employing the dedicated probing signal. And the improvement increases with the number of CUs.

\begin{figure}[t]
\centering
  \includegraphics[width=0.5\textwidth]{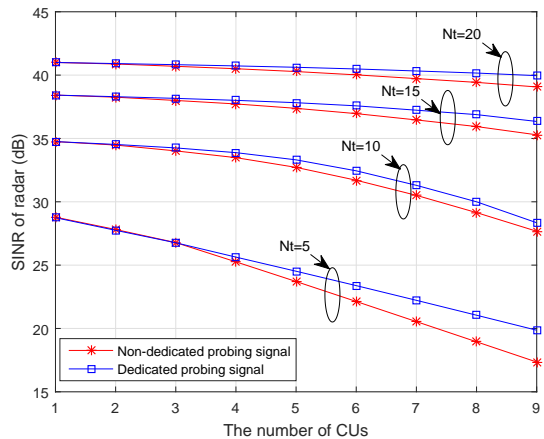}
\caption{Average SINR of radar with different number of CUs and antennas, $N_t=N_r$, $\Gamma=25$ dB.  }
\label{KSNR}
\end{figure}

\section{Conclusion}
In this paper, we have proposed the joint optimization of transmit and receive beamforming for the DFRC system. The optimal tradeoff of SINR between radar and communication has been characterized through maximizing the SINR of the radar under the SINR constraints of the CUs. For the single CU scenario, we have given the closed-form solution of the optimized beamforming, and it has been proved that there is no need of dedicated probing signals. For the multiple CUs scenario, the beamforming design has been formulated as a non-convex QCQP. We have obtained the optimal solutions by applying SDR with rand-1 property, and it has been proved that the dedicated probing signal should be employed to improve the SINR of the radar. Numerical results have been provided to show that our algorithm is effective. 

Future research may focus on beamforming designs with further practical constraints, e.g. constant modulus constraints. Also, our work is based on the prior information of both radar and communication. Thus, the robust design will be considered with the imperfect prior information of the radar and the imperfect CSI of the communication.

\appendices
\section{Proof of Proposition 1} \label{proof1}
If $\Gamma  \le {\left| {{{\bf{h}}^H}{\bf{\hat g}}} \right|^2}$, the SINR constraint of the CU is inactive. Thus, Problem (P1.2) reduces to the MIMO radar beamforming optimization problem without communication constraints, i.e., 

\begin{equation}
\begin{array}{*{20}{l}}
{\mathop {\max }\limits_{\bf{u}} }&{\bar \gamma _R^{\left( \text{I} \right)} = {\bf{u}}_{}^H{{\bf{\Phi }}_0}{\bf{u}}}\\[3mm]
{{\rm{s}}{\rm{.t}}{\rm{.}}}&{{\bf{u}}_{}^H{\bf{u}} \le {P_0}}
\end{array},
\end{equation}
and the corresponding optimal solution is ${\bf{u}}_{}^* = \sqrt {{P_0}} {\bf{\hat g}}$.

If ${P_0}{\left| {\bf{h}} \right|^2} \ge \Gamma  > {P_0}{\left| {{{\bf{h}}^H}{\bf{\hat g}}} \right|^2}$, the SINR constraint is active. The optimal beamforming vector should lie in the pace jointly spanned  by ${\bf{\hat h}}$  and the projection of ${\bf{g}}$  into the null space of ${\bf{\hat h}}$, i.e., ${{\bf{\hat g}}_ \bot }$, where ${\bf{g}}_ \bot ^H{\bf{\hat h}} = {\bf{0}}$.
The transmit power of ${\bf{u}}_{}^*$  allocated in the direction of ${\bf{\hat h}}$  should satisfy the SINR constraint with the coefficient $\alpha$ in (\ref{eq16}), and the left power is allocated in the direction of ${{\bf{\hat g}}_ \bot }$ with the coefficient $\beta$ in (\ref{eq17}), which improves the SINR of the radar receiver without influencing the SINR of the CU. 
 
If $\Gamma  > {P_0}{\left| {\bf{h}} \right|^2}$, the SINR constraint cannot be satisfied even with maximum ratio transmission to the CU.

\section{Proof of Proposition 2} \label{proof2}
Assuming the power of the dedicated probing signal is $\tau {P_0}$  with  $0 \le \tau  \le 1$, Problem (P2.2) can be decomposed into the following two problems, i.e.,

\begin{equation}
\left( {{\rm{P2}}{\rm{.3}}} \right)\begin{array}{*{20}{l}}
{\mathop {\max }\limits_{\bf{v}} }&{{{\bf{v}}^H}{{\bf{\Phi }}_0}{\bf{v}}}\\[3mm]
{{\rm{s}}{\rm{.t}}{\rm{.}}}&{{{\bf{v}}^H}{\bf{v}} \le \tau {P_0}}
\end{array}
\end{equation}
and
\begin{equation}
\left( {{\rm{P2}}{\rm{.4}}} \right)\begin{array}{*{20}{l}}
{\mathop {\max }\limits_{\bf{u}} }&{{\bf{u}}_{}^H{{\bf{\Phi }}_0}{\bf{u}}}\\[3mm]
{{\rm{s}}{\rm{.t}}{\rm{.}}}&{\bar \gamma _C^{} ={\left| {{{\bf{h}}^H}{\bf{u}}} \right|^2} \ge \Gamma }\\[1.5mm]
{}&{{{\bf{u}}^H}{\bf{u}} \le \left( {1 - \tau } \right){P_0}}
\end{array}.
\end{equation}

For Problem (P2.3), the optimal beamforming vector of the probing signal can be calculated as

\begin{equation}
\label{eq29}
{{\bf{v}}^*} = \sqrt {\tau {P_0}} {\bf{\hat g}},
\end{equation}
where ${\bf{\hat g}} = {{\bf{g}} \mathord{\left/
 {\vphantom {{\bf{g}} \left\| {\bf{g}} \right\|}} \right.
 \kern-\nulldelimiterspace} \left\| {\bf{g}} \right\|}$ and ${\bf{g}}$ is the dominant eigenvector of  ${{\bf{\Phi }}_0}$. 
For Problem (P2.4), assuming $P_1=\tau P_0$, the optimal beamforming vector of the probing signal can be calculated  according to Proposition 1, i.e.,

\begin{equation}
{\bf{u}}_{}^* = \left\{ {\begin{array}{*{20}{l}}
{\sqrt {\left( {1 - \tau } \right){P_0}} {\bf{\hat g}},}&{\Gamma  \le \left( {1 - \tau } \right){P_0}{{\left| {{{\bf{h}}^H}{\bf{\hat g}}} \right|}^2}}\\[3mm]
{\left( {\alpha '{\bf{\hat h}} + \beta '{{{\bf{\hat g}}}_ \bot }} \right),}&\begin{array}{l}
\left( {1 - \tau } \right){P_0}{\left| {\bf{h}} \right|^2} \ge \Gamma \\
\Gamma  > \left( {1 - \tau } \right){P_0}{\left| {{{\bf{h}}^H}{\bf{\hat g}}} \right|^2}
\end{array}
\end{array}} \right.,
\end{equation}
\begin{equation}
\alpha ' = \sqrt {\frac{\Gamma }{{{\left\| {\bf{h}} \right\|^2}}}} \frac{{{\alpha _g}}}{{\left| {{\alpha _g}} \right|}},
\end{equation}
\begin{equation}
\beta '= \sqrt {(1-\tau){P_0} - \frac{\Gamma }{{{\left\| {\bf{h}} \right\|^2}}}} \frac{{{\beta _g}}}{{\left| {{\beta _g}} \right|}},
\end{equation}
where ${\bf{g}}$ is the dominant eigenvector of  ${{\bf{\Phi }}_0}$, ${\bf{\hat g}} = {{\bf{g}} \mathord{\left/
 {\vphantom {{\bf{g}} \left\| {\bf{g}} \right\|}} \right.
 \kern-\nulldelimiterspace} \left\| {\bf{g}} \right\|}$, ${\bf{\hat h}} = {{{{\bf{h}}}} \mathord{\left/
 {\vphantom {{{{\bf{h}}}} {\left\| {\bf{h}} \right\|}}} \right.
 \kern-\nulldelimiterspace} {\left\| {\bf{h}} \right\|}}$, 
${{\bf{g}}_ \bot } = {\bf{g}} - ( {{{{\bf{\hat h}}}^H}{\bf{g}}} ){\bf{\hat h}}$ denoting the projection of ${\bf{g}}$  into the null space of  ${\bf{\hat h}}$, ${{\bf{\hat g}}_ \bot } = {{{\bf{g}}_ \bot } \mathord{\left/
 {\vphantom {{{\bf{g}}_ \bot } {\left\| {{\bf{g}}_ \bot ^{}} \right\|}}} \right.
 \kern-\nulldelimiterspace} {\left\| {{\bf{g}}_ \bot ^{}} \right\|}}$ and ${\bf{ g}}$  can be expressed as ${\bf{ g}} = {\alpha _g}{\bf{\hat h}} + {\beta _g}{{\bf{\hat g}}_ \bot }$.

According to (\ref{eq29}), one has

\begin{equation}
{\bf{u}}_{}^* + {{\bf{v}}^*} = \left\{ {\begin{array}{*{20}{l}}
{\sqrt {{P_0}} {\bf{\hat g}},}&{\Gamma  \le \left( {1 - \tau } \right){P_0}{{\left| {{{\bf{h}}^H}{\bf{\hat g}}} \right|}^2}}\\[3mm]
{\left( {\alpha '{\bf{\hat h}} + \beta '{{{\bf{\hat g}}}_ \bot }} \right),}&\begin{array}{l}
\left( {1 - \tau } \right){P_0}{\left| {\bf{h}} \right|^2} \ge \Gamma \\
\Gamma  > \left( {1 - \tau } \right){P_0}{\left| {{{\bf{h}}^H}{\bf{\hat g}}} \right|^2}
\end{array}
\end{array}} \right.,
\end{equation}
Thus, ${{\bf{u}}^ * } + {{\bf{v}}^*}$ can achieve the best performance when $\tau=0$. That is the optimal beamforming vector of the probing signal ${{\bf{v}}^*} = {\bf{0}}$  for Problem (P2.2), which completes the proof. 

\section{Proof of Proposition 3} \label{proof3}
Consider a separable semidefinite program
(SDP) as follows:

\begin{equation}
\begin{array}{*{20}{l}}
{\mathop {\min }\limits_{\left\{ {{{\bf{X}}_l}} \right\}} }&{\sum\limits_{l = 1}^L {{\rm{tr}}\left( {{{\bf{B}}_l}{{\bf{X}}_l}} \right)} }\\[3mm]
{{\rm{s}}{\rm{.t}}{\rm{.}}}&{\sum\limits_{l = 1}^L {{\rm{tr}}\left( {{{\bf{C}}_{ml}}{{\bf{X}}_l}} \right){ \triangleright _m}{b_m},m = 1, \cdots ,M} }\\[1.5mm]
{}&{{{\bf{X}}_l} \succ 0,l = 1, \cdots ,L}
\end{array},    
\end{equation}
where $\left\{ {{{\bf{B}}_l}} \right\}$, $\left\{ {{{\bf{C}}_{ml}}} \right\}$ are all Hermitian matrices (not necessarily positive semidefinite), ${b_m} \in \mathcal{R},\forall m$, ${ \triangleright _m} \in \left\{ { \ge , \le , = } \right\},\forall m$, and $\left\{ {{{\bf{X}}_l}} \right\}$ are all Hermitian matrices. 

Suppose the above SDP is feasible and bounded, where the optimal value is attained. Then, according to \cite{5233822}, it always has an optimal solution $\left\{ {{\bf{X}}_l^*} \right\}$ such that 

\begin{equation}
\sum\limits_{l = 1}^L {{{\left[ {{\rm{rank}}\left( {{\bf{X}}_l^*} \right)} \right]}^2}}  \le M. 
\end{equation}
Based on the above result, it can be proved that there is always a solution to Problem (P3.3)
satisfying that

\begin{equation}
\label{eqb62}
\sum\limits_{k = 1}^K {{{\left[ {{\rm{rank}}\left( {{\bf{U}}_k^*} \right)} \right]}^2}}  \le K + 1.
\end{equation}
Meanwhile, due to the SINR constraints of each CU, one has ${\bf{U}}_k^* \ne 0, \forall k$ and then ${\rm{rank}}\left( {{\bf{U}}_k^*} \right) \ge 1, \forall k$.
Thus, ${\rm{rank}}\left( {{\bf{U}}_k^*} \right) = 1, \forall k$ according to (\ref{eqb62}), which completes the proof.

\section{Proof of Proposition 4} \label{proof4}
According to the rank constraints of the SDP in Appendix \ref{proof3} of \cite{5233822}, it can be proved that there is always a solution to Problem (P3.3)
satisfying that 

\begin{equation}
\label{eqb63}
\sum\limits_{k = 1}^K {{{\left[ {{\rm{rank}}\left( {{\bf{U}}_k^*} \right)} \right]}^2}}  + {\left[ {{\rm{rank}}\left( {{{\bf{V}}^*}} \right)} \right]^2} \le K + 1.
\end{equation}
Meanwhile, due to the SINR constraints of each CU,  ${\bf{U}}_k^* \ne 0, \forall k$, and then ${\rm{rank}}\left( {{\bf{U}}_k^*} \right) \ge 1, \forall k$.
Thus, one has ${\rm{rank}}\left( {{\bf{U}}_k^*} \right) = 1, \forall k$ and ${\rm{rank}}\left( {{{\bf{V}}^*}} \right) \le 1$ according to (\ref{eqb63}).

Next, we prove the necessity of employing the dedicated probing signal. The Lagrangian function of Problem (P4.3) is

\begin{equation}
{{\cal L}_3} = \mu {P_0} - \sum\limits_{k = 1}^K {{\lambda _k}}  + \sum\limits_{k = 1}^K {{\rm{tr}}} \left( {{{\bf{D}}_k}{{\bf{U}}_k}} \right) + {\rm{tr}}\left( {{\bf{EV}}} \right),
\end{equation}
where 

\begin{equation}
{{\bf{D}}_k} = {{\bf{\Phi }}_0} + \frac{{{{\bf{H}}_k}}}{{{\Gamma _k}}} - \sum\limits_{j \ne k} {{\lambda _j}{\rm{tr}}\left( {{{\bf{H}}_j}} \right)}  - \mu {\bf{I}},    
\end{equation}
and

\begin{equation}
{\bf{E}} = {{\bf{\Phi }}_0} - \mu {\bf{I}},  
\end{equation}
with ${\lambda _k} \ge 0,\forall k$ and $\mu  \ge 0$ being the dual variables associated with the SINR constraint of the CU $k$ and the transmit power constraint, respectively. And the dual problem is

\begin{equation}
\begin{array}{*{20}{l}}
{\mathop {\min }\limits_{\left\{ {{\lambda _k}} \right\},\mu } }&{\mu {P_0} - \sum\limits_{k = 1}^K {{\lambda _k}} }\\[3mm]
{{\rm{s}}{\rm{.t}}{\rm{.}}}&{{{\bf{D}}_k} \prec 0,\forall k,{\bf{E}} \prec 0}
\end{array}
\end{equation}

Note that Problem (P4.3) is convex. Thus, strong duality holds and the KKT conditions are necessary and sufficient for any optimal solution to Problem (P4.3). Due to the fact that ${\bf{E}} = {{\bf{\Phi }}_0} - \mu {\bf{I}} \preceq 0$. 
We have that ${\mu ^*} \ge \kappa $, where $\kappa $  is the dominant eigenvalue of ${{\bf{\Phi }}_0}$. 
If  ${\mu ^*} = \kappa $, we have that ${\bf{v}}_{}^* = \sqrt {\tau {P_0}} {\bf{\hat g}}$,
with  $0 < \tau  \le 1$. 
If ${\mu ^*} > \kappa $, we have that  ${\bf{v}}_{}^* = {\bf{0}}$. It completes the proof.

\bibliography{ref.bib}{}
\bibliographystyle{IEEEtran}
\end{document}